\documentclass[draftcls,onecolumn,romanappendices]{IEEEtran}
\usepackage[latin1]{inputenc}
\usepackage[english]{babel}
\usepackage{graphicx, subfigure}
\usepackage{amsmath, amsfonts, amssymb, latexsym}
\usepackage{float,latexsym, basicsymbols, notation}
\usepackage{cite}
\usepackage{color}

\usepackage{stackrel}
\usepackage{epstopdf}

\newtheorem{thm}{Theorem}
\newtheorem{lem}{Lemma}

\newtheorem{cor}{Corollary}

\newcommand{\refE}[1] {(\ref{eqn:#1})}

\newcommand{\refF}[1] {Fig.~\ref{fig:#1}}

\newcommand{\Csis}{Csisz\'ar}

\IEEEoverridecommandlockouts
\begin{document}

%\title{A Derivation of the Error Exponent\\ of Source-Channel Coding}
\title{A Derivation of the Source-Channel Error Exponent using Non-identical Product Distributions}
%\title{Attaining the Source-Channel Error Exponent with Non-identical Product Distributions}

%\author{\authorblockN{Adri\`a Tauste Campo,
%Gonzalo Vazquez-Vilar,
%Albert Guill{\'e}n i F{\`a}bregas,
%Tobias Koch and
%Alfonso Martinez}
%\authorblockA{$^1$University of Cambridge, $^2$Universitat Pompeu Fabra, $^3$Instituci\'o Catalana de Recerca i Estudis Avan\c{c}ats (ICREA)}
%\authorblockA{Email: \{atauste,gvazquez,guillen,alfonso.martinez\}@ieee.org, tobi.koch@eng.cam.ac.uk }

\author{Adri\`a Tauste Campo,
Gonzalo Vazquez-Vilar,
Albert Guill{\'e}n i F{\`a}bregas,\\
Tobias Koch and
Alfonso Martinez
\thanks{A. Tauste Campo, G. Vazquez-Vilar and A. Martinez are with the Department of Information and Communication Technologies, Universitat Pompeu Fabra, Barcelona, Spain (emails: \{atauste,gvazquez,alfonso.martinez\}@ieee.org).
A. Guill\'en i F\`abregas is with the Instituci\'o Catalana de Recerca i Estudis Avan\c{c}ats (ICREA), the Department of Information and Communication Technologies, Universitat Pompeu Fabra, Barcelona, Spain, and the Department of Engineering, University of Cambridge, CB2 1PZ Cambridge, United Kingdom (email: guillen@ieee.org).
T. Koch is with the Signal Theory and Communications Department, Universidad Carlos III de Madrid, 28911 Legan\'es, Spain (email: koch@tsc.uc3m.es).
%T. Koch was with the Department of Engineering, University of Cambridge, CB2 1PZ Cambridge, United Kingdom. He is now with the Signal Theory and Communications Department, Universidad Carlos III de Madrid, 28911 Legan\'es, Spain (email: koch@tsc.uc3m.es).
}
\thanks{
This work has been funded in part by the European Research Council (ERC)
under grant agreement 259663; by the European Union under the 7th Framework Programme, grants FP7-PEOPLE-2009-IEF no. 252663, FP7-PEOPLE-2011-CIG no. 303633, FP7-PEOPLE-2012-CIG no. 333680, FP7-PEOPLE-2013-IEF no. 329837; and by the Spanish Ministry of Economy and Competitiveness under grants CSD2008-00010, TEC2009-14504-C02-01, TEC2012-38800-C03-01, TEC2012-38800-C03-03 and RYC-2011-08150. A. Tauste Campo acknowledges funding from an EPSRC (Engineering and Physical Sciences Research Council, UK) Doctoral Prize Award.}
\thanks{This work was presented in part at the 46th Conference on Information Sciences and Systems, Princeton, NJ, March 21--23, 2012 and at the IEEE Symposium on Information Theory, Cambridge, MA, July 1--6, 2012.}}
\maketitle 

\thispagestyle{empty}

\begin{abstract}
This paper studies the random-coding exponent of joint source-channel coding for a scheme where source messages are assigned to disjoint subsets (referred to as classes), and codewords are independently generated according to a distribution that depends on the class index of the source message. For discrete memoryless systems, two optimally chosen classes and product distributions are found to be sufficient to attain the sphere-packing exponent in
those cases where it is tight.
\end{abstract}

%%%%%%%%%%%%%%%%%%%%%%%%%%%%%%%%%%%%%%%%%%%%%%%%%%%%%%%%%%
\section{Introduction}\label{sec:intro}
%%%%%%%%%%%%%%%%%%%%%%%%%%%%%%%%%%%%%%%%%%%%%%%%%%%%%%%%%%
%In \cite{Shannon48}, Shannon proved  the source-channel coding theorem
%for stationary memoryless sources and channels. 
%The direct part of the theorem states that a source of entropy $H(V)$ can be transmitted over a channel of capacity $C$ with vanishing error probability as the block length grows large  if $H(V)<C$. Conversely, the error probability is bounded away from zero if $H(V)>C$. 
%For the achievability part, Shannon used separate source-channel coding,  
%indirectly showing that the concatenation of source and channel codes suffices to asymptotically achieve vanishing error probability. 
%Yet, for a fixed block length, 

Jointly designed source-channel codes may achieve a lower error probability than separate source-channel coding \cite{Csis80}. In fact, the error exponent of joint design may be up to twice that of the concatenation of source and channel codes \cite{Zhong06}. The best exponent in this setting is due to {\Csis} \cite{Csis80}, who used a construction where codewords are drawn at random from a set of sequences with a composition that depends on the source message. He also showed that the exponent coincides with an upper bound, the sphere-packing exponent, in a certain rate region.

Gallager \cite[Prob.~5.16]{Gall68} derived a random-coding exponent for an  ensemble whose codewords are drawn according to a fixed product distribution, independent of the source message. %; and do not necessarily have a fixed type.
This method %, which naturally extends to channels with continuous alphabets and memory, 
yields a simple derivation of the channel coding exponent in
discrete memoryless channels \cite[Th. 5.6.2]{Gall68}.
However, the straightforward application to source-channel coding gives a (generally)
weaker achievable exponent than \Csis's method, although this difference is typically small for the optimum choice of input distributions~\cite{Zhong06}.
%, the methods used to derive each exponent are conceptually different, which raises the question of whether the difference lies in the the composition of codewords (fixed-composition vs. product-distribution), in the ensemble choice (varying codeword distribution for different source messages vs. identically distributed codewords), or in the bounding technique of the average error probability (method of types vs. Gallager's techniques). This can be summarized in a number of questions: 
%\begin{enumerate}
%\item Can the sphere-packing exponent be attained with random codes generated by
%product distributions?
%\item Do codeword distributions need to be source-message-dependent?
%\item Do Gallager's bounding techniques suffice to derive \Csis's exponent?
%\item Does the formula for the best exponent hold beyond discrete memoryless systems?
%\end{enumerate}

In this paper, %we answer the first three questions in the positive by highlighting the importance of the ensemble choice. Specifically, 
we study a code ensemble for which codewords associated to different source messages are generated according to different product distributions. We derive a new random-coding bound on the error probability for this ensemble and show that its exponent %for both product-distribution and fixed-composition codes can attain
attains the sphere-packing exponent in the cases where it is tight. 
%However, the codewords associated to different source messages may need to be generated according to different distributions.
%To show this, we derive a good upper bound to the average error probability of an ensemble that encompasses product-distribution and fixed-composition ensembles and apply Gallager's bounding techniques to derive a good upper bound on the ensemble-average error probability. 
%In this case, %We then find that the exponential rate of decay of this bound coincides with \Csis's exponent when 
We find that either one or two different distributions suffice in the optimum ensemble.
%Some of our results naturally extend to channels with continuous alphabets and source-channel pairs with memory, partly answering the fourth question.

%Finally, we analytically characterize the gap in the exponent between Gallager's and \Csis's exponent for source-channel  memoryless systems as
%\begin{align}
% %\LEJG&=
%& \max_{\rho \in[0,1]} \inf_{R>0}\left\{  \Eo(\rho, \pch)- \rho R + t\bar{F}\left(\frac{R}{t},\Delta, \ps\right)\right\}\\
%\leq &\inf_{R>0} \left\{   \max_{\rho \in[0,1]}\left\{\bar{\Eo}(\rho , \pch)-\rho R\right\} +tF\left(\frac{R}{t},\Delta, \ps\right)\right\},% \label{minmax_ineq}\\
%%&=\LEJCs,
%\end{align}
%where $t$ is the transmission rate of the system, $\Eo(\rho, \pch)$ is the  Gallager's channel function \footnote{The function is here optimized over the channel input distribution} for the channel law $\pch$, $F\left(\frac{R}{t},\Delta, \ps\right)$ 
%is the source reliability function for a distribution $\ps$  and distortion threshold $\Delta$, and $\bar{\Eo}(\cdot )$ and $\bar{F}(\cdot)$ denote the concave hull of $\Eo(\cdot )$ over $\rho\in[0,1]$ and 
%the convex hull of $F(\cdot)$ over $R>0$ respectively. 

The paper is structured as follows. 
In Section \ref{sec:intro.notation} we introduce the system model
and several definitions used throughout the paper.
Section \ref{Section:Previous_work} reviews related previous work on
source-channel coding. Section \ref{Section:Random_coding_bound},
the main section of the paper, presents the new random-coding
bound and its error exponent.
Finally, we conclude in Section \ref{Section:conclusions} with
some final remarks. 
Proofs of the results can be found in the appendices.

\section{System Model and Definitions}\label{sec:intro.notation}

An encoder maps a source message $\v$ to a length-$n$ codeword $\x(\v)$, which is then transmitted over the channel and decoded as $\hat{\v}$ at the receiver upon observation of the output $\y$.
The source is characterized by a distribution $\Ps(\v)=\prod_{j=1}^{k}\ps(v_j)$, $\v=(v_1,\dots, v_k)\in \Va^k$, where $\Va$ is a finite alphabet. 
Since $P$ fully describes the source, we shall sometimes abuse notation and refer to $P$ as the source.
The channel law is given by a conditional probability distribution $\Pch(\y|\x)=\prod_{j=1}^{n}\pch(y_j|x_j)$, $\x=(x_1,\dots, x_n)\in \Xa^n$,
$\y=(y_1,\dots, y_n)\in \Ya^n$, where $\Xa$ and $\Ya$ denote the input and output alphabet, respectively.
While $\Xa$ and $\Ya$ are assumed discrete for ease of exposition, our achievability results  extend in a natural way to continuous alphabets.

Based on the output $\y$, the decoder selects a source message $\hat{\v}$ according to the maximum a posteriori (MAP) criterion,
\begin{align}\label{eqn:MAPdec}
  \hat{\v} = \arg \max_{\v} \Ps(\v)\Pch\bigl(\y | \x(\v)\bigr).
\end{align}
Here and throughout the paper, we avoid explicitly writing
the set in optimizations and summations if they are performed
over the entire set. Also, where unambiguous, we shall write
$\x$ instead of $\x(\v)$.
We study the average error probability $\epsilon$, defined as
\begin{equation}\label{error_joint}
 \epsilon\triangleq \Pr\{\hat{\V}\neq \V \},
 \end{equation}
where capital letters are used %here, and throughout the paper, 
to denote random variables. %Throughout the paper, we shall use the notation $\A\sim P_{\A}$ to indicate that $\A$ is distributed according to the distribution $P_{\A}$.
In addition to bounds on the average error probability $\epsilon$ for finite values of $k$ and $n$, we are interested in its exponential decay. Consider a sequence of sources with length $k=1,2,\ldots$ and a corresponding
sequence of codes of length $n=n_1,n_2,\ldots$
Assume that the ratio $\frac{k}{n}$ converges to some quantity 
\begin{equation}\label{eqn:tdef}
  t\triangleq\lim_{k\to\infty} \frac{k}{n},
\end{equation}
referred to as \emph{transmission rate}.
An exponent $E(P,W,t) > 0$ is to said to be achievable if there exists a sequence
of codes whose error probabilities $\epsilon$ satisfy
\begin{equation}\label{jscc_exponent_def}
\epsilon \leq e^{-n E(P,W,t)+o(n)}, 
\end{equation}
where $o(n)$ is a sequence such that $\lim_{n\to\infty} o(n)/n=0$. The reliability function $E_{\text{J}}(P,W,t)$ is defined as the supremum of
all achievable error exponents; we sometimes shorten it to $E_{\text{J}}$.

We denote Gallager's source and channel functions as
\begin{align}\label{Gallager_exponent_PX_2}
 \Es(\rho, \ps) &\triangleq
    \log \left(\sum_{v} \ps(v)^{\frac{1}{1+\rho}} \right)^{1+\rho},\\
  \Eo(\rho,\pch,\px) &\triangleq
   -\log \sum_{y}  \left( \sum_{x} \px(x)
       \pch(y|x)^{\frac{1}{1	+\rho}}\right)^{1+\rho},
\end{align}
respectively.

Sometimes, we are interested in the error exponent maximized only over a subset of probability distributions on $\Xc$. Let $\Qc$ be a non-empty proper subset of probability distributions on $\Xc$. With some abuse of notation we define
\begin{equation}\label{Gallager_exponent_optPx}
  \Eo(\rho, \pch, \Qc) \triangleq \max_{\px \in \Qc} \Eo(\rho, \pch, \px).
\end{equation}
When the optimization is done over the set of all
probability distributions on $\Xc$ we simply write
$\Eo(\rho, \pch) \triangleq \max_{\px} \Eo(\rho, \pch, \px)$.

We denote by $\barEo(\rho, \pch, \Qc)$ the concave hull of $\Eo(\rho, \pch, \Qc)$,
defined pointwise as the supremum  over all convex combinations
of any two values of the function $\Eo(\rho,\pch,\Qc)$~\cite[p. 36]{convex-rockafellar},
i.e.
\begin{equation}\label{eqn:barEodef}
  \bar\Eo(\rho, \pch, \Qc) \triangleq
   \max_{\substack{\rho_1,\rho_2,\lambda\in[0,1]:\\ \lambda\rho_1+(1-\lambda)\rho_2=\rho}}
   \Bigl\{ \lambda \Eo(\rho_1, \pch, \Qc) + (1-\lambda) \Eo(\rho_2, \pch, \Qc) \Bigr\}.
\end{equation}
Similarly, we write $\barEo(\rho, \pch)$ to denote the concave hull of
$\Eo(\rho, \pch)$.

%%%%%%%%%%%%%%%%%%%%%%%%%%%%%%%%%%%%%%%%%%%%%%%%%%%%%%%%%%%%%%%%%%%%%%%%%%%
\section{Previous Work: Gallager's and {\Csis}'s Exponents}\label{Section:Previous_work}
%%%%%%%%%%%%%%%%%%%%%%%%%%%%%%%%%%%%%%%%%%%%%%%%%%%%%%%%%%%%%%%%%%%%%%%%%%%

For source coding (i.e., when $W$ is the channel law of a noiseless channel),
the reliability function of a source $\ps$ at rate $R$, denoted by $e(R, \ps)$, is given by  \cite{Jel68}%\cite{Blahut74,Jel68}
\begin{align}
e(R, \ps) &= \sup_{\rho\geq 0} \bigl\{\rho R - \Es(\rho, \ps)\bigr\}.\label{source_exponent_dual}
\end{align}
%with $ D(\cdot\|\cdot)$ denoting the divergence between two distributions, $Q\sim P_{Q}$ being a dummy random variable, and 

%In the interval $H(V)\leq R \leq t \log |\Va|$, \eqref{source_exponent}
%becomes \cite[eq. (7)]{Csis80}
%\begin{equation}
%e(R, \ps) = \min_{Q:H(Q)=R} D(P_{Q}\|P)\label{source_exponent2}.
%\end{equation}

%$\bar\Eo\bigl(\rho,\pch,\Qc\bigr)$ denotes the concave hull over $\rho$ of the function $\Eo\bigl(\rho,\pch,\Qc\bigr)$.

For channel coding (i.e., when $P$ is the uniform distribution), the
reliability function of a channel $\pch$ at rate $R$, denoted by $E(R, \pch)$,
is bounded as~\cite{Gall68}
\begin{equation}
\Er(R, \pch) \leq E(R, \pch)\leq \Esp(R, \pch),
\end{equation}
where  $\Er(R, \pch)$ is the random-coding exponent and $\Esp(R, \pch)$ is the sphere-packing exponent, respectively, given by
 \begin{align}
&\Er(R, \pch) %&\triangleq  %\max_{\px}\min_{V} \left\{D(V\|\pch|\px)+|I(V,\px) -R|^{+}\right\}\\
    \triangleq \max_{\rho\in[0,1]} \,\Bigl\{\Eo(\rho, \pch)-\rho R \Bigr\}, \label{channel_exponent} \\
&\Esp(R, \pch) %&\triangleq  %\max_{\px}\min_{V} \left\{D(V\|\pch|\px)+|I(V,\px) -R|^{+}\right\}\\
                                   \triangleq \sup_{\rho\geq 0}\, \Bigl\{ \Eo(\rho, \pch)-\rho R \Bigr\}. \label{spherepacking_exponent}
\end{align}

For source-channel coding Gallager used a random-coding argument to derive an upper bound on the average error probability by drawing the codewords independently of the source messages according to a given product distribution
$\Px(\x)= \prod_{j=1}^{n} \px (x_j)$.
He found the achievable exponent~\cite[Prob. 5.16]{Gall68}
\begin{equation}\label{Gallager_exponent_Px}
\max_{\rho\in[0,1]}\, \Bigl\{ \Eo(\rho, \pch, \px)-t\Es(\rho, \ps)\Bigr\},
\end{equation}
which becomes, upon maximizing over $\px$,
\begin{equation}\label{Gallager_exponent}
\EJG(P,W,t) \triangleq \max_{\rho\in[0,1]}\, \Bigl\{ \Eo(\rho, \pch)-t\Es(\rho, \ps)\Bigr\}.
\end{equation}

{\Csis} refined this result using the method
of types \cite{Csis80}.
By using a partition of the message set into source-type classes and considering fixed-composition codes that map messages within a source type onto sequences within a channel-input type,
%Secondly, a suboptimal maximum mutual information decoder
%is used  at the receiver. This decoder first decides on
%the source type that is being transmitted and then on the
%source message within the type. Finally, {\Csis} uses a channel-coding result for messages with unequal error protection \cite[Th. 5]{Csis80} to prove that for every $\delta>0$, there exists an  $n_0\in \NN$ %\left(\delta,|\Xa|, |\Ya|, |\Va|\right)\in \NN$ 
%such that, for $n\geq n_0$, %\left(\delta,|\Xa|, |\Ya|, |\Va|\right)$, 
%the probability of error ${\epsilon}$ can be upper-bounded as
%\begin{equation}\label{Csis_error_bound}
%{\epsilon}\leq \sum_{i=1}^{N_k}e^{-k e\left(\frac{n R_i}{k}, \ps\right) -n\Er(R_i, \pch)-2\delta},
%{\epsilon}\leq \sum_{i=1}^{N_k}e^{-n\bigl(te\left(\frac{R_i}{t}, %\ps\right) +\Er(R_i, \pch)-2\delta\bigr)},
%\end{equation}
%
%\begin{IEEEeqnarray}{lCl}
%\EJ & \geq & \EJCs \nonumber\\
%& \triangleq & \min_{H(V)\leq R\leq t\log \abs{\Va}} \biggl\{
%t e\left(\frac{R}{t},\ps\right)+\Er(R, \pch)\biggr\},\IEEEeqnarraynumspace \label{Csiszar_exponent}
%\end{IEEEeqnarray}
%provided that $n\geq n_0\left(\delta,|\Xa|, |\Ya|, |\Va|\right)$, 
%where $N_k$ denotes the number of source-type classes.
%$e(R, \ps)$ is ,
%and where $\Er(R, \pch)$ is the channel random-coding
%exponent, given by \cite{Gall68}
%where  $I(V,P_X)$ denotes the mutual information when the channel
%law is  $V$ and the channel-input distribution
%is $P_X$. 
%When $\lim_{n\to\infty} \frac{k}{n} = t$,
%Eq.~\eqref{Csis_error_bound} 
he found an achievable exponent
\begin{IEEEeqnarray}{lCcCl}
\EJCs(P,W,t) %\\
& \triangleq &  \min_{t H(V)\leq R\leq R_{\Vcal}} \biggl\{
t e\left(\frac{R}{t},\ps\right)+\Er(R, \pch)\biggr\},\IEEEeqnarraynumspace \label{Csiszar_exponent}
\end{IEEEeqnarray}
where $R_{\Vcal} \triangleq t \log |\Vcal|$. A convenient alternative representation of $\EJCs$ was obtained by Zhong \emph{et al.} \cite{Zhong06}
via Fenchel's duality theorem \cite[Thm. 31.1]{convex-rockafellar}:
%, namely, which allows one to rewrite \eqref{Csiszar_exponent} as
\begin{equation}\label{Concave_Hull_exponent}
\EJCs(P,W,t) = \max_{\rho\in[0,1]} \bigl\{\barEo(\rho, \pch) - t \Es(\rho, \ps)\bigr\}.
\end{equation}
Since $\barEo(\rho, \pch) \geq \Eo(\rho, \pch)$, it follows from \eqref{Concave_Hull_exponent} and \eqref{Gallager_exponent} that $\EJCs \geq \EJG$ in general. Nonetheless, the finite-length bound implied by the exponent $\EJCs$ in \cite{Csis80} might be worse than the one in \cite[Prob. 5.16]{Gall68} due to the worse subexponential terms, which may dominate for finite values of $k$ and $n$.

To validate the optimality of $\EJCs$,
{\Csis} derived a sphere-packing  bound on the exponent \cite[Lemma 2]{Csis80}, 
\begin{IEEEeqnarray}{lCl}
\EJSp(P,W,t) \triangleq \min_{t H(V)\leq R\leq R_{\Vcal}} \biggl\{
t e\left(\frac{R}{t},\ps\right)+\Esp(R, \pch)\biggr\}.\IEEEeqnarraynumspace \label{Csiszar_sp_exponent}
\end{IEEEeqnarray}
When the minimum on the right-hand side (RHS) 
of \eqref{Csiszar_sp_exponent} is attained  for a value of $R$ such that $\Esp(R, \pch)=\Er(R, \pch)$, the upper bound \eqref{Csiszar_sp_exponent} coincides with the lower bound \eqref{Csiszar_exponent} and, hence, $\EJCs=\EJ$. This is the case for values of $R$ above the critical rate of the channel $\Rcr$ \cite{Csis80}.

%%%%%%%%%%%%%%%%%%%%%%%%%%%%%%%%%%%%%%%%%%%%%%%%%%%%%%%%%%%%%%%%%%%%%%%%%%%%%%
\section{An Achievable Exponent for Joint Source-Channel Coding}
\label{Section:Random_coding_bound}

In this section, we analyze the error probability of random-coding ensembles where the codeword distribution depends on the source message. 
We find that ensembles generated with a pair of product distributions $\bigl\{\Px_{1},\Px_{2}\bigr\}$ may attain a better error exponent than Gallager's exponent~\eqref{Gallager_exponent_Px} for $\px$ being equal to either $\px_{1}$ or $\px_{2}$. Moreover, optimizing over pairs of distributions this ensemble recovers the exponent $\EJSp$
in those cases where it is tight.

\subsection{Main Results}
\label{Section:Results}

Let us first define a partition of the source-message set
$\Va^k$ into $N_k$ disjoint  subsets $\Acal_k^{(i)}$,  $i=1,\ldots,N_k$, such that $\bigcup_{i=1}^{N_k} \Acal_k^{(i)} = \Va^k$. We refer to these subsets as \emph{classes}. 
For each source message $\v$ in the set $\Acal_k^{(i)}$, we randomly and independently generate codewords $\x(\v) \in \Xa^n$ according to a channel-input  product distribution $\Px_{i}(\x)= \prod_{j=1}^{n} \px_i (x_j)$.
This definition is a generalization of \Csis's  partition in \cite{Csis80}
where each subset corresponds to a source-type class. Since the number of
source-type classes is a polynomial function of $k$~\cite{Csis98}, it
follows that the number of classes $N_k$ considered in \cite{Csis80}
is also polynomial in $k$.

The next result extends \cite[Th. 5.6.2]{Gall68} to codebook ensembles
where codewords are independently but not necessarily identically distributed.
%For the general partition we have the following theorem.

%%%%%%%%%%%%%%%%%%%%%%%%%%%%%%%%%%%%%%%%%%%%%%%%%%%%%%%%%%%%%%%
\begin{thm}\label{Main_theorem}
For a given partition $\Acal_k^{(i)}$,  $i=1,\ldots,N_k$, and
associated distributions $\px_{i}$, $i=1,\ldots,N_k$, there exists
a codebook satisfying
\begin{IEEEeqnarray}{lCl}
\epsilon &\leq &
%\epsilon_{\text{B}}\left(\Pcal_k\right) \triangleq 
               \hfunc \sum_{i=1}^{N_{k}}
               \exp\biggl(-\max_{\rho_{i}\in[0,1]}\Bigl\{\Eo\bigl(\rho_i,\Pch,\Px_{i}\bigr)  -\EsI(\rho_i,\Ps)\Bigr\}\biggr), \label{main_error_bound}
\end{IEEEeqnarray}
where $\hfunc \triangleq \frac{3N_k-1}{2}$ and %the function $\EsI(\rho,\Ps)$ is given by
%We also define the  functions
\begin{align} \label{eqn:def_Eis}
  \EsI(\rho,\Ps) \triangleq \log \left( \sum_{\v \in \Acal_k^{(i)}} \Ps(\v)^{\frac{1}{1+\rho}} \right)^{1+\rho}.
\end{align}

\end{thm}
\begin{proof} See Appendix \ref{Main_theorem_proof}.
\end{proof}

%\begin{remark}
Theorem \ref{Main_theorem} holds for general (not necessarily memoryless)
discrete sources and channels, and for $\Px_{i}$, $i=1,\ldots,N_k$, being
non-product distributions (including cost-constrained and fixed composition
ensembles).
Furthermore, it naturally extends to continuous channels by following the
same arguments as those extending Gallager's exponent for channel coding.  
In particular, it can be generalized beyond the scope of \cite{Zhong07} and
\cite{Zhong09}, where Markovian sources and Gaussian channels were studied,
respectively.
%\end{remark}

It was demonstrated in~\cite{jsccisit12} that an application of Theorem~\ref{Main_theorem} to a partition where classes are identified with source-type classes attains $\EJCs$.
However, compared to the bound used to derive {\Csis}'s exponent in~\cite{Csis80},
Theorem~\ref{Main_theorem} provides a tighter bound on the average error
probability for finite values of $k$ and $n$~\cite{jsccciss2012}.
Along different lines, Theorem~\ref{Main_theorem} can be generalized to derive
\Csis's lower bound on the error exponent for lossy source-channel coding~\cite{Csis82}.

For a single class with associated distribution $\px$,
Theorem~\ref{Main_theorem} simply recovers the exponent in~\eqref{Gallager_exponent_Px}.
%If we can choose from several input distributions, the best overall exponent is obtained by taking the highest exponent, which is in turn equivalent to maximizing over the corresponding $\Eo(\rho_1,\pch,\px)$ functions. Surprisingly, a larger exponent can be obtained if we partition the set of source messages into two classes, as we see next. 
The following theorem shows that the exponent may be improved by
considering a partition with two classes.
\begin{thm}\label{thm:Cs2classes}
For a pair of distributions $\{\px, \px'\}$, there exists a partition
of the source message set into two classes such that the following
exponent is achievable
\begin{equation} \label{eqn:thmCs2classes}
\max_{\rho\in[0,1]}\,\Bigl\{\bar{\Eo}\bigl(\rho,\pch,\{\px, \px'\}\bigr) -t \Es(\rho,\ps)\Bigr\}.
\end{equation}
Moreover, a partition achieving this exponent is given by
 \begin{align}
   \Acal_k^{(1)}(\gamma)&\triangleq \left\{ \v: \; \Ps(\v) < \gamma^k \right\}\label{eqn:A1def}\\
   \Acal_k^{(2)}(\gamma)&\triangleq \left\{ \v: \; \Ps(\v) \geq \gamma^k \right\},\label{eqn:A2def}
\end{align}
for some $\gamma\in[0,1]$ with associated distributions
$\px_{i} \in \bigl\{\px, \px'\bigr\}$, $i=1,2$.
\end{thm}
\begin{proof}
See Appendix \ref{sec:Cs2classes_dem1}.
\end{proof}

In Theorem \ref{thm:Cs2classes} we considered a particular
pair of distributions $\{\px, \px'\}$.
%A direct application of Carath\'eodory's theorem shows that
%any point $(x,\bar f(x))$ belonging to the graph of the
%concave hull of a function $f(x) = \max_{\lambda\in\Lambda} f_{\lambda}(x)$
%can be expressed as a convex combination of two points
%belonging to the graph of $f(x)$.
%If we choose a pair $\lambda_1, \lambda_2 \in \Lambda$
%such those points belong to the graph of $\max \left\{ f_{\lambda_1}(x), f_{\lambda_2}(x)\right\}$ then $(x,\bar f(x))$ also belongs to the
%graph of $g(x) = \max_{\px\in\{\px_{\alpha},\px_{\beta}\}} {\Eo}(\rho,\pch,\px)$,
%then $(x,\bar f(x))=(x,\bar g(x))$.
%"Convexity" by H. G. Eggleston,
%Cambridge tracts in mathematics and mathematical physics, page 35.
%Theorem 18 (ii) (Fenchel-Eggleston): If in addition, A is the union of
%at most n connected sets, then the number of dimensions needed is \leq
%n.
A direct application of Carath\'eodory's theorem~\cite[Cor. 17.1.5]{convex-rockafellar}
shows that any point belonging to the graph of $\bar{\Eo}(\rho,\pch)$
can be expressed as a convex combination of two points belonging to
the graph of $\Eo(\rho,\pch)$.
Consequently, there exists a pair of distributions $\px, \px'$ such
that these two points also belong to the graph of $\Eo(\rho,\pch,\{\px, \px'\})$.
By optimizing the exponent \refE{thmCs2classes} over all possible
pairs of distributions $\{\px, \px'\}$, the following result follows.

\begin{cor}\label{cor:Cs2classes}
There exists a partition of the source message set into two
classes assigned to a pair of distributions such that
$\EJCs$ in \eqref{Concave_Hull_exponent}
%\begin{equation} \label{eqn:corCs2classes}
%\EJCs(P,W,t) = \max_{\px, \px'}
%\max_{\rho\in[0,1]}\,\Bigl\{\bar{\Eo}\bigl(\rho,\pch,\{\px, \px'\}\bigr) -t \Es(\rho,\ps)\Bigr\}
%\end{equation}
is achievable.
\end{cor}

In contrast to Csisz\'ar's original analysis \cite{Csis80},
where the number of classes used to attain the best exponent
was polynomial in $k$, Corollary~\ref{cor:Cs2classes}
shows that a two-class construction suffices to attain $\EJCs$
when the partition and associated distributions are appropriately
chosen.

\subsection{Ensemble Tightness}\label{sec:discussion}

We have studied the error probability of random-coding ensembles where different
codeword distributions are assigned to different subsets of source messages.
Since Section~\ref{Section:Results} only considers achievability results,
one may ask whether the weakness of Gallager's exponent is due to the
bounding technique or the construction itself. A partial answer to this question
can be given by studying the exact random-coding exponent, namely the exact
exponential decay of the error probability averaged over the  ensemble, which
we denote by $\bar{\epsilon}$.

%In the remainder of this section we consider the partition of $\Va^k$ into
%source-type classes~\cite{CsisKorn11}.
%Note that finer partitions of the source message set will not yield a larger
%error exponent.
%For messages $\v$ in the $i$-th type class, codewords are drawn according
%to the product distribution $\Px_{i}$, for some $\px_i \in \Qc$, $i=1,\ldots,N_k$,
%where $\Qc$ is a non-empty set of probability distributions on $\mathcal{X}$.

\begin{thm}\label{thm:EnsembleTightness}
For any non-empty set $\Qc$ of probability distributions on $\mathcal{X}$,
consider a codebook ensemble for which the codewords associated to
source messages with type class $\Tcal_i$ are generated according to
a distribution $\Px_{i}(\x)= \prod_{j=1}^{n} \px_i (x_j)$ with
$\px_i \in \Qc$, $i=1, \dotsc, N'_k$, where $N'_k$ is the number
of source type classes.
The random-coding exponent of this ensemble is upper-bounded as
\begin{align} \label{thm_ub_eqn}
\limsup_{n\to \infty}-\frac{\log\bar{\epsilon}}{n}
  &\leq \max_{\rho\in[0,1]} \bigl\{ \bar{\Eo}(\rho, \pch, \Qc)-t \Es(\rho, \ps)\bigr\}.
\end{align}
\end{thm}
\begin{proof}
See Appendix \ref{proof_thm_ub}.
\end{proof}

When $\Qc$ contains only one distribution, the concavity of $\Eo(\rho, \pch, \px)$
as a function of $\rho$ shows that the RHS of \eqref{thm_ub_eqn} matches~\eqref{Gallager_exponent_Px}. In other words, if the
codebook is drawn according to only one distribution $\px$, then
$\EJG$ in \eqref{Gallager_exponent} cannot be improved, i.e.,
it is ensemble tight.

The ensemble considered in Theorem \ref{thm:Cs2classes} is a particular
case of that of Theorem \ref{thm:EnsembleTightness} with $|\Qc| = 2$.
Since the upper bound \eqref{thm_ub_eqn} and the lower bound \refE{thmCs2classes}
coincide for $\Qc = \{\px, \px'\}$, the error exponent \refE{thmCs2classes}
is also ensemble tight.
Furthermore, for any set with cardinality $\Qc$ with $|\Qc| > 2$,
we can always choose two distributions $\px$ and  $\px'$ belonging to
$\Qc$ such that \refE{thmCs2classes} equals the RHS of
\eqref{thm_ub_eqn}~\cite[Cor. 17.1.5]{convex-rockafellar}.
Therefore, the random-coding exponent of an ensemble
with an arbitrary number of classes can be attained by the
two-class partition proposed in Theorem \ref{thm:Cs2classes}.

Finally, it can be shown that Theorem \ref{thm:EnsembleTightness} holds
for finer partitions of the source message set, not necessarily corresponding
to source type classes. 
Since the RHS of~\eqref{thm_ub_eqn} coincides with $\EJCs$ when
$\mathcal{Q}$ is the set of all probability distributions on
$\mathcal{X}$, we conclude that the ensembles studied in this
work cannot improve \Csis's random-coding exponent, even when the
latter does not coincide with the sphere-packing exponent.
%In this case, other ensembles could improve the bound by introducing dependence
%over codeword pairs. For example, the expurgation technique \cite{Gall68}
%eliminates bad pairs of codewords from the code.
%Nevertheless, this problem lies beyond the scope of this paper.

\begin{figure}[t]
  \begin{center}
  \includegraphics[width=0.9\columnwidth]{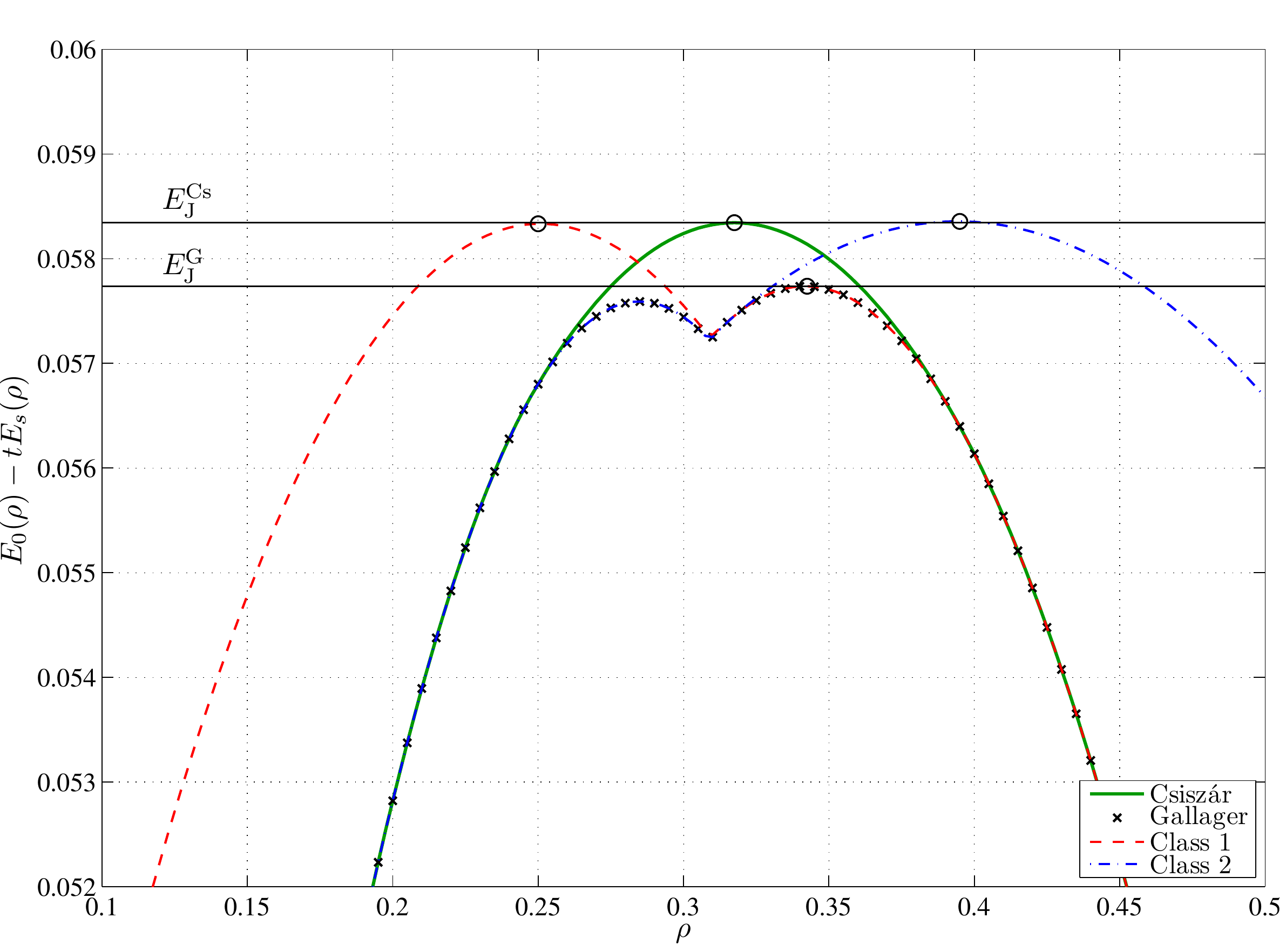}
	\vspace{-5mm}\caption{Error exponent bounds. Csisz\'ar's and Gallager's
	curves correspond to $\barEo(\rho, \pch) - t \Es(\rho, \ps)$ and $\Eo(\rho, \pch) - t \Es(\rho, \ps)$, respectively. Class $i$ curve correspond to $\Eo(\rho, \pch) \; - \stackrel[{n\to\infty}]{}{\lim} \frac{1}{n} \EsI(\rho, \Ps)$, for $i=1,2$.}%when $\gamma^{\star}=$ \bf{TO BE SET (GONZALO)}.} %$\rateTh = \rateTh^{\star}$.}
	\label{fig:exp_vs_rho}
  \end{center}
  \vspace{-5mm}
\end{figure}

\subsection{Example: a 6-input 4-output channel}\label{sec:example}
%%%%%%%%%%%%%%%%%%%%%%%%%%%%%%%%%%%%%%%%%%%%%%%%%%%%%%%%%%%%%%%%%%%%%%%%%%%%%%%%
We present an example\footnote{In this subsection all logarithms and
exponentials are computed to base $2$. Hence all the information quantities
related to this example are expressed in bits.} in which the two-class partition (with their corresponding product distributions) attains the sphere-packing exponent
while Gallager's one-class assignment does not.
Consider the source-channel pair composed by a binary memoryless source (BMS)
and a non-symmetric memoryless channel with $|\Xc|=6$, $|\Yc|=4$ and
transition-probability matrix
\begin{align}\label{Example}
\pch =
   \left(\begin{array}{c c c c}
    1-3\xi_1 & \xi_1 & \xi_1 & \xi_1\\
    \xi_1 & 1-3\xi_1 & \xi_1 & \xi_1\\
    \xi_1 & \xi_1 & 1-3\xi_1 & \xi_1\\
    \xi_1 & \xi_1 & \xi_1 & 1-3\xi_1\\
    \frac{1}{2}-\xi_2 & \frac{1}{2}-\xi_2 & \xi_2 & \xi_2\\
    \xi_2 & \xi_2 & \frac{1}{2}-\xi_2 & \frac{1}{2}-\xi_2\\    
        \end{array} \right).
\end{align}
This channel is similar to the channel given in \cite[Fig. 5.6.5]{Gall68}
and studied in \cite{Zhong06} for source-channel coding.
It is composed of two quaternary-output sub-channels:
one of them is a quaternary-input symmetric channel with parameter
$\xi_1$, and the second one is a binary-input channel with parameter $\xi_2$.
We set $\xi_1=0.065$, $\xi_2=0.01$, $t=2$ and $\ps(1) = 0.028$.
It follows that the source entropy is $H(V)=0.1843$ bits/source symbol, the channel capacity is $C=0.9791$ bits/channel use and the critical rate is $\Rcr=0.4564$ bits/channel use.
Let $R^{\star}$ denote the value of $R$ minimizing \eqref{Csiszar_exponent}.
In this example we have $R^{\star}=0.6827>\Rcr$ and $\EJCs$ is tight.

In \refF{exp_vs_rho} we plot the objective functions of Gallager's exponent in
\eqref{Gallager_exponent} and \Csis's exponent in \eqref{Concave_Hull_exponent}
as functions of $\rho$, respectively.
For reference purposes, we also show the values of $\EJG$ and $\EJCs$
with horizontal solid lines. 
The distribution $\px$ maximizing $E_0(\rho,\pch,\px)$ changes from
$\bigl(\frac{1}{4}\; \frac{1}{4}\; \frac{1}{4}\; \frac{1}{4}\; 0\; 0\bigr)$
for $\rho \leq 0.31$ to $\bigl(0\; 0\; 0\; 0\; \frac{1}{2}\; \frac{1}{2}\bigr)$
for $\rho > 0.31$.
As a result, $E_0(\rho,\pch)$ is not concave in $\rho\in[0,1]$.
The figure shows how the non-concavity of Gallager's function around
the optimal $\rho$ of Csisz\'ar's function translates into a loss in
exponent. 

\refF{exp_vs_rho} also shows the bracketed terms 
in the RHS of~\eqref{main_error_bound} as a function of $\rho_i$
for the two-class partition of Theorem~\ref{thm:Cs2classes}.
The overall error exponent of the two-class construction is obtained
by first individually maximizing the exponent of each of the curves
over $\rho_i$, and by then choosing the minimum of the two individual
maxima.
In this example, the exponent of both classes coincides with $\EJCs$.
The overall exponent is thus given by $\EJCs$, which is in agreement
with Theorem \ref{thm:Cs2classes}.

\section{Conclusions}\label{Section:conclusions}

We have studied the error probability of random-coding ensembles where different
codeword distributions are assigned to different subsets of source messages.
We have showed that the random-coding exponent of ensembles generated with a
single distribution does not attain \Csis's exponent in general.
In contrast, ensembles with at most two appropriately chosen subsets
and distributions suffice to attain the sphere-packing exponent in
those cases where it is tight.
One of the strengths of our achievability result is that, unlike \Csis's
approach, it does not rely on the method of types. This leads to tighter
bounds on the average error probability for finite block lengths and
may simplify the task of generalizing our bound to source-channel systems
with non-discrete alphabets and memory.

% on the JSCC random-coding error probability
%that recovers Gallager's and \Csis's lower bounds on the error exponent for discrete
%memoryless systems. Thus, the new expression gives the actual error exponent
%at least in the cases where \Csis's exponent is tight.
%
%The method to obtain the new bound uses a specific random-coding construction with
%MAP decoding. Specifically, we partition the message set into disjoints classes and assign to each class an input distribution according to which the codewords are randomly generated.
%
%By partitioning the message set into source-type classes, and by choosing for each class the input distribution to be a product distribution, the new bound on the error probability recovers \Csis's lower bound on the error exponent, answering the question of whether fixed composition codes are required to achieve \Csis's exponent in the negative.
%
%One of the strengths of Gallager's error bound \eqref{Gall_error_bound} is that it applies also to channels with memory, and that it can be easily generalized to nondiscrete channels. While in the derivation of our new bound we used mostly the same techniques as Gallager, there are some steps that rely on the method of types. Consequently, our new bound lacks the above mentioned strengths of Gallager's bound. Nevertheless, unlike \Csis's approach, our approach makes use of the method of types only in the analysis of the error probability, and not in the encoder and the decoder. This may simplify the task of generalizing our bound to channels with memory or to nondiscrete channels.

\appendices

\section{Proof of Theorem \ref{Main_theorem}} \label{Main_theorem_proof}
%%%%%%%%%%%%%%%%%%%%%%%%%%%%%%%%%%%%%%%%%%%%%%%%%%%%%%%%%%%%%%%%%
%The derivation of the error bound from Theorem \ref{Main_theorem}
%relies in the intermediate result presented in the next lemma.
%\begin{lem}\label{lem:Main_theorem_lemma}
%For arbitrary  $s_i \in \bigl[\frac{1}{2},1\bigr]$, $i=1,\ldots,N_k$,
%we have the following upper bound on the probability of error:
%\begin{IEEEeqnarray}{lCl}\label{eqn:Main_theorem_lemma}
%\epsilon &\leq & \sum_{i,j=1}^{N_k} \sum_{\y} G_i(\y)^{s_i} G_j(\y)^{1-s_i}
%\end{IEEEeqnarray}
%where for $i=1,\ldots,N_k$ we define
%\begin{IEEEeqnarray}{lCl}
%G_i(\y) &\triangleq & \Biggl( \sum_{\v\in \Acal_k^{(i)}} \Ps(\v)^{s_i}\Biggr)^{\frac{1}{s_i}} \Biggr(\sum_{\x} \Px_{i}(\x) \Pch(\y|\x)^{s_i} \Biggr)^{\frac{1}{s_i}}. \IEEEeqnarraynumspace \label{eqn:gi-def}
%\end{IEEEeqnarray}
%\end{lem}
%

Generalizing the proof of the random-coding union bound for channel coding
\cite[Th. 16]{Pol09} (with earlier precedents in \cite[pp.~136-137]{Gall68})
to the cases where codewords are independently generated according to
distributions that depend on the class index of the source, we obtain
%the partition we consider  gives the following upper bound to the average error probability:
%\begin{equation}\label{RCU_JSCC}
% \bar{\epsilon} \leq \sum_{\v,\x,\y} \Ps(\v)\Pxv(\x|\v)\Pch(\y|\x)
%\\ \times
% \min\Biggl\{1,\sum_{\bar{\v}\neq \v}
%\Pr\left\{ {\Ps(\bar{\v})\Pch(\y| \bar{\X})} \geq {\Ps(\v)\Pch(\y|\x)}
%\right\}
% \Biggr\},
%\end{equation}
%where
% the probability on the RHS of \eqref{RCU_JSCC} is computed with respect to $\bar{\X}\sim P_{\X|\V}(\X|\V=\bar{\v})$
%and the expectation is jointly taken over the variables $(\V,\X,\Y) \sim \Ps \Pxv \Pch$.
%
%To prove Theorem~\ref{Main_theorem}, we apply the RCU bound  \eqref{RCU_JSCC} together with the proposed code construction to obtain
\begin{IEEEeqnarray}{lCl}
\epsilon & \leq & \sum_{i=1}^{N_k} \sum_{\v\in \Acal_k^{(i)}} \Ps(\v)\sum_{\x,\y} \Px_{i}(\x) \Pch(\y|\x)
%\notag\\& & \qquad \qquad  {} \times %%%%%%%%%%%%%%%%%%%%%%%%%%%%%%%%%%%%%
 \min\left\{1, \sum_{j=1}^{N_k} \sum_{\bar{\v}\in \Acal_k^{(j)}} \sum_{\substack{\bar{\x}: \Ps(\bar{\v})\Pch(\y|\bar{\x}) \\ \geq \Ps(\v)\Pch(\y|\x)}} \Px_{j}(\bar{\x})\right\}.\IEEEeqnarraynumspace \label{eqn:mt-proof-0}
\end{IEEEeqnarray}

We next use Markov's inequality for $s_j \geq 0$, $j=1,\ldots,N_k$, to obtain~\cite{Gall68}
\begin{equation}\label{eqn:mt-proof-1}
\sum_{\substack{\bar{\x}: \Ps(\bar{\v})\Pch(\y|\bar{\x}) \\ \geq \Ps(\v)\Pch(\y|\x)}} \Px_{j}(\bar{\x}) \leq \sum_{\bar{\x}} \Px_{j}(\bar{\x})
           \Biggl(\frac{\Ps(\bar{\v})\Pch(\y|\bar{\x})}
                {\Ps(\v)\Pch(\y|\x)}\Biggr)^{s_j}.
\end{equation}

Using \eqref{eqn:mt-proof-1} and the inequality $\min\{1,A+B\} \leq A^{\rho}+B^{\rho'}$, $A,B\geq 0$, $\rho,\rho' \in [0,1]$ \cite{Gall68}, \refE{mt-proof-0} is upper-bounded by
\begin{IEEEeqnarray}{lCl}
{\epsilon} & \leq & \sum_{i,j=1}^{N_k}
 \sum_{\v\in \Acal_k^{(i)}}
                 \Ps(\v)\sum_{\x,\y} \Px_{i}(\x) \Pch(\y|\x)
\notag\\& & \qquad \qquad  {} \times 
 \left(\sum_{\bar{\v}\in \Acal_k^{(j)}}
   \sum_{\bar{\x}} \Px_{j}(\bar{\x})
           \Biggl(\frac{\Ps(\bar{\v})\Pch(\y|\bar{\x})}
                {\Ps(\v)\Pch(\y|\x)}\Biggr)^{s_j}
\right)^{\rho_{ij}}, \IEEEeqnarraynumspace \label{eqn:cross-class-error-event-1}
\end{IEEEeqnarray}
where $\rho_{ij} \in [0,1]$ and $s_j \geq 0$, $i,j=1,\ldots,N_k$.

For $s_i,s_j\in \bigl[\frac{1}{2},1\bigr]$ and $\rho_{ij}=\frac{1-s_i}{s_j}$, \refE{cross-class-error-event-1} yields
%Choosing $\rho_{ij} = \frac{1-s_i}{s_j} \in [0,1]$ for $s_i,s_j \in \bigl[\frac{1}{2},1\bigr]$ in \refE{cross-class-error-event-1}, rearranging the resulting expression, and turning back into \refE{epsboundepsij} 
%and applying the definition of $G_i(\y)$, we obtain the desired result.
%\end{proof}
%yields
\begin{IEEEeqnarray}{lCl}\label{eqn:Main_theorem_lemma}
\epsilon & \leq & \sum_{i,j=1}^{N_k} \sum_{\y} G_i(\y)^{s_i} G_j(\y)^{1-s_i}
\end{IEEEeqnarray}
where
\begin{IEEEeqnarray}{lCl}
G_i(\y) &\triangleq & \left( \sum_{\v\in \Acal_k^{(i)}} \Ps(\v)^{s_i}\right)^{\frac{1}{s_i}} \Biggr(\sum_{\x} \Px_{i}(\x) \Pch(\y|\x)^{s_i} \Biggr)^{\frac{1}{s_i}}. \IEEEeqnarraynumspace \label{eqn:gi-def}
\end{IEEEeqnarray}
%Lemma \ref{lem:Main_theorem_lemma} partitions the error bound
%depending on the class of the transmitted codeword and on the
%class of the (wrong) decoded codeword. We now show that the
%probability of the inter-class error events can be bounded by
%that of the intra-class error events.

This choice of $\rho_{ij}$ allows us to decompose the probability of
the ``inter-class'' error event between classes $i$ and $j$ as the
product of two terms corresponding to the ``intra-class'' error
events of each class.
The RHS of \refE{gi-def} is further upper-bounded by
\begin{IEEEeqnarray}{lCl}
\epsilon %&\leq & \sum_{i,j=1}^{N_k} \sum_{\y} G_i(\y)^{s_i} G_j(\y)^{1-s_i}
%  \IEEEeqnarraynumspace \label{eqn:cross-class-error-event-3}\\
   &\leq & \sum_{i,j=1}^{N_k} \left( \sum_{\y} G_i(\y) \right)^{s_i} \left( \sum_{\y} G_j(\y) \right)^{1-s_i}
  \IEEEeqnarraynumspace \label{eqn:cross-class-error-event-4}\\
   &\leq & \sum_{i,j=1}^{N_k} \left( s_i \left( \sum_{\y} G_i(\y) \right) + (1-s_i) \left( \sum_{\y} G_j(\y) \right) \right)\label{eqn:cross-class-error-event-5}\\
   &\leq & \sum_{i=1}^{N_k} \sum_{\y} G_i(\y) + \sum_{\substack{i,j=1\\i\neq j}}^{N_k} \left( \sum_{\y} G_i(\y) + \frac{1}{2} \sum_{\y} G_j(\y)\right)
  \IEEEeqnarraynumspace \label{eqn:cross-class-error-event-6}\\
   &= & \frac{3N_k - 1}{2} \sum_{i=1}^{N_k} \sum_{\y} G_i(\y),
  \IEEEeqnarraynumspace \label{eqn:cross-class-error-event-7}
\end{IEEEeqnarray}
where in \refE{cross-class-error-event-4} we applied H\"older's
inequality $\|fg\| \leq \|f\|_{p} \|g\|_{q}$ with $p=\frac{1}{s_i}$ and
$q=\frac{1}{1-s_i}$;
%such that $p,q \geq 1$ and $\frac{1}{p}+\frac{1}{q}=1$;
\refE{cross-class-error-event-5} follows from the relation
between arithmetic and geometric means; %and in
and \refE{cross-class-error-event-6} follows because % we used the bounds 
$\frac{1}{2} \leq s_i \leq 1$.
%in the terms of the sum for which $i \neq j$.
%Then, result follows from \refE{cross-class-error-event-7}
%by using the definition of $G_i(\y)$, $i=1,\ldots,N_k$.
By identifying 
%s_i= 1/(1+\rho_i)
%\rho_i= (1-s_i)/s_i
\begin{equation} 
\sum_{\y} G_i(\y) = \exp\Biggl( - \Eo\left(\frac{1-s_i}{s_i},\Pch,\Px_{i}\right) + \EsI\left(\frac{1-s_i}{s_i},\Ps\right) \Biggr) 
\end{equation} 
and optimizing over $\frac{1}{2} \leq s_i \leq 1$, $i=1,\ldots,N_k$,
it follows that
%Finally, by using the definition of $G_i(\y)$, $i=1,\ldots,N_k$, it
%follows from  \refE{cross-class-error-event-7} that
\begin{equation}
\epsilon \leq 
               \frac{3N_k - 1}{2}\sum_{i=1}^{N_{k}}
               \exp\biggl(-\max_{\rho_{i}\in[0,1]}\Bigl\{\Eo\bigl(\rho_i,\Pch,\Px_{i}\bigr)  -\EsI(\rho_i,\Ps)\Bigr\}\biggr),
\end{equation}
where we denote $\frac{1-s_i}{s_i}$ by $\rho_i$.
This concludes the proof.

\section{Proof of Theorem~\ref{thm:Cs2classes}}
%%%%%%%%%%%%%%%%%%%%%%%%%%%%%%%%%%%%%%%%%%%%%%%%%%%%%%%%%%%%%%%%%%%%%%%%%%%%%%%%%%
\label{sec:Cs2classes_dem1}

%$\Acal_k^{(1)} = \left\{ \v: \; \Ps(\v) < \gamma^k \right\}$ and $\Acal_k^{(2)} = \left\{ \v: \; \Ps(\v) \geq \gamma^k \right\}$.
%\begin{align}
%        \Acal_k^{(1)}(\gamma) \triangleq \left\{ \v: \; \Ps(\v) \geq \gamma^k \right\},\label{eqn:def_Acal1}\\
%        \Acal_k^{(2)}(\gamma) \triangleq \left\{ \v: \; \Ps(\v) < \gamma^k \right\},\label{eqn:def_Acal2}
%\end{align}

%Consider the family of partitions \refE{A1def}-\refE{A2def}
%parametrized by $\gamma \in [0,1]$. This family does not
%change if we allow the threshold 
%By allowing the threshold
%to cover the range of non-negative reals  space $\gamma \g \in [0,1]$
%
The proof of the Theorem~\ref{thm:Cs2classes} is based on the next 
preliminary result. %that applies to the two-class partition of the statement.
%that will be denoted by
%\begin{align}
%   \Acal_k^{(1)}(\gamma) \triangleq \left\{ \v: \; \Ps(\v) < \gamma^k \right\}\label{eqn:A1def}\\
%   \Acal_k^{(2)}(\gamma) \triangleq \left\{ \v: \; \Ps(\v) \geq \gamma^k \right\}.\label{eqn:A2def}
%\end{align}

\begin{lem}\label{lem:Esibound} %%%%%%%%%%%%%%%%%%%%%%%%%%%%%%%%%%%%%%%%%%%%%%%%%%
%Let $\gamma= \min\{1,\gamma'\}$
%such that 
For any $\rho_0\in[0,1]$ and $\gamma'\geq 0$, the partition
\refE{A1def}-\refE{A2def} with $\gamma=\min\{1,\gamma'\}$
satisfies
\begin{align} 
    \frac{1}{k}\Es^{(1)}(\rho,\Ps)
      &\leq \Es(\rho,\ps) \openone\{ \rho > \rho_0 \}
        + r(\rho,\rho_0,\gamma') \openone\{ \rho \leq  \rho_0 \} 
          \triangleq \bar\Es^{(1)}(\rho,\rho_0,\gamma'),
    \label{eqn:Es1bound}\\
    \frac{1}{k}\Es^{(2)}(\rho,\Ps)
      &\leq \Es(\rho,\ps) \openone\{ \rho < \rho_0 \}
        + r(\rho,\rho_0,\gamma') \openone\{ \rho \geq  \rho_0 \}
          \triangleq \bar\Es^{(2)}(\rho,\rho_0,\gamma'),
    \label{eqn:Es2bound}
\end{align}
where $\openone\{\cdot\}$ denotes the indicator function, and where
\begin{align} 
  r(\rho,\rho_0,\gamma) \triangleq \Es(\rho_0,\ps) 
    + \frac{\Es(\rho_0,\ps) - \log\gamma}{1+\rho_0} \left( \rho-\rho_0 \right).
\label{eqn:rdef}
\end{align}
\end{lem}
\begin{proof}
For the choice $\gamma=\min\{1,\gamma'\}$ it holds that
\begin{align}
\openone\left\{  \Ps(\v) < \gamma^k \right\} 
  \leq \openone\left\{  \Ps(\v) \leq \gamma^k \right\}
  =    \openone\left\{  \Ps(\v) \leq (\gamma')^k \right\}
  \label{eqn:Es1eq0}
\end{align}
since $\Ps(\v)\leq 1$ for all $\v$. Using \refE{Es1eq0} 
and the bound $\openone\{a \leq b\} \leq a^{-s} b^s$
for $s\geq0$, the function $\frac{1}{k}\Es^{(1)}(\rho,\Ps)$
can be upper-bounded as
\begin{align} 
  \frac{1}{k}\Es^{(1)}(\rho,\Ps)
%  &= \frac{1}{k} \log \left( \sum_{\v} \Ps(\v)^{\frac{1}{1+\rho}}
%                   \openone\left\{  \Ps(\v) < \gamma^k \right\} \right)^{1+\rho}
%    \label{eqn:Es1eq1}\\
%  &\leq \frac{1}{k} \log \left( \sum_{\v} \Ps(\v)^{\frac{1}{1+\rho}}
%                   \openone\{  \Ps(\v) \leq \gamma^k \} \right)^{1+\rho}
%    \label{eqn:Es1eq2}\\
  &\leq \frac{1}{k} \log \left( \sum_{\v} \Ps(\v)^{\frac{1}{1+\rho}}
                   \openone\left\{  \Ps(\v) \leq (\gamma')^k \right\} \right)^{1+\rho}
    \label{eqn:Es1eq2b}\\
  &\leq \frac{1}{k} \log \left( \sum_{\v} \Ps(\v)^{\frac{1}{1+\rho}} 
                   \Ps(\v)^{-s} (\gamma')^{ks} \right)^{1+\rho}
    \label{eqn:Es1eq3}\\
  &= \log \left(  \sum_{v} \ps(v)^{\frac{1}{1+\rho}-s}
                   (\gamma')^{s} \right)^{1+\rho},
    \label{eqn:Es1eq4}
\end{align}
for any $s \geq 0$. Here we used that $\Ps(\v)$ is memoryless.
We continue by choosing $s$ such that
\begin{align}
  s &= \max\left(0, \frac{\rho_0-\rho}{(1+\rho_0)(1+\rho)}\right).
    \label{eqn:sstar}
\end{align}
For $\rho >\rho_0$, it then follows that $s=0$, and
%from 
%\refE{Es1eq1}-
\refE{Es1eq4} gives (cf. \eqref{Gallager_exponent_PX_2})
%we obtain
\begin{align} 
  \frac{1}{k}\Es^{(1)}(\rho,\Ps) &\leq \Es(\rho,\ps).
  \label{eqn:psi1eq0}
\end{align}
For $\rho \leq  \rho_0$,  the choice \refE{sstar} yields $s=\frac{\rho_0-\rho}{(1+\rho_0)(1+\rho)}$, which together
with \refE{Es1eq4} yields
\begin{align} 
  \frac{1}{k}\Es^{(1)}(\rho,\Ps) &\leq
     (1+\rho)\log\left(\sum_{v} \ps(v)^{\frac{1}{1+\rho_0}}\right) 
                          - \frac{\rho-\rho_0}{1+\rho_0} \log\gamma'
\label{eqn:psi1eq2}\\
    &= (1+\rho_0)\log\left(\sum_{v} \ps(v)^{\frac{1}{1+\rho_0}}\right) 
      +(\rho-\rho_0)\log\left(\sum_{v} \ps(v)^{\frac{1}{1+\rho_0}}\right) 
                          - \frac{\rho-\rho_0}{1+\rho_0} \log\gamma'
\label{eqn:psi1eq3}\\
    &= \Es(\rho_0,\ps)
     + \frac{ \Es(\rho_0,\ps)
        - \log\gamma'}{1+\rho_0} \left( \rho-\rho_0 \right),
    \label{eqn:psi1eq4}
%\label{eqn:psi1eq4}\\
%     \quad = \quad  r(\rho,\rho_0,\gamma)
\end{align}
where in \refE{psi1eq3} we added and subtracted the term
$\rho_0 \log\left(\sum_{v} \ps(v)^{\frac{1}{1+\rho_0}}\right)$;
and \refE{psi1eq4} follows from the definition \eqref{Gallager_exponent_PX_2}.
The inequality \refE{Es1bound} follows by  combining \refE{psi1eq0} and \refE{psi1eq2}-\refE{psi1eq4} for
$\rho> \rho_0$ and $\rho\leq \rho_0$, respectively.

In an analogous way,  the inequality \refE{Es2bound} can be proved
%The inequality \refE{Es1bound} from the lemma follows by substituting \refE{psi1eq0} and \refE{psi1eq2}-\refE{psi1eq4} in each of the regions of $\rho$.
using that $\openone\{  \Ps(\v) \geq \gamma^k \} = \openone\{  \Ps(\v) \geq (\gamma')^k \}$ and 
$\openone\{a \geq b\} \leq a^s b^{-s}$ with $s \geq 0$.
%the inequality \refE{Es2bound} can be proved %shown 
%following the very same steps.
\end{proof} %%%%%%%%%%%%%%%%%%%%%%%%%%%%%%%%%%%%%%%%%%%%%%%%%%%%%%%%%%%%%%%%

By applying Theorem~\ref{Main_theorem} to the two-class partition
\refE{A1def}-\refE{A2def} with associated product distributions
$\Px_i$, $i=1,2$, for the optimal threshold $\gamma$ we obtain
\begin{align}
  \EJB
  &\triangleq \max_{\gamma\in[0,1]} \left\{ \liminf_{n\to\infty}
    \left\{ - \frac{1}{n} \log \Biggl(  h(k)
\sum_{i=1,2} e^{- \stackrel[{\rho_i\in[0,1]}]{}{\max} \bigl\{n\Eo(\rho_i,\pch,Q_i) - \EsI(\rho_i) \bigr\} } \Biggr) \right\} \right\}\label{eqn:proof2_2classes0}\\
         &= \max_{\gamma \in[0,1]} \left\{ \liminf_{n\to \infty}
          \min_{i=1,2}  \left\{ \max_{\rho_i \in[0,1]}  \left\{ \Eo(\rho_i,\pch,Q_i) - \frac{1}{n} \EsI(\rho_i)\right\} \right\} \right\}\label{eqn:proof2_2classes1}\\
         &\geq \max_{\gamma'\geq 0} \max_{\rho_0,\rho_1,\rho_2 \in[0,1]}
               \min_{i=1,2}  \left\{
         \Eo(\rho_i,\pch,Q_i) - t \bar\Es^{(i)}(\rho_i,\rho_0,\gamma') \right\} \label{eqn:proof2_2classes2}\\
         &\geq \max_{\substack{\rho_0,\rho_1,\rho_2 \in[0,1]:\\\rho_1\leq\rho_0\leq\rho_2}}
               \max_{\gamma'\geq 0}
               \min_{i=1,2}  \left\{ \Eo(\rho_i,\pch,Q_i) - t \bar\Es^{(i)}(\rho_i,\rho_0,\gamma')  \right\}, \label{eqn:proof2_2classes3}
\end{align}
where \refE{proof2_2classes1} follows by noting that $h(k)$ is
subexponential in $k$; in \refE{proof2_2classes2} we have applied
Lemma \ref{lem:Esibound} with $\rho_0 \in [0,1]$ and $\gamma' \geq 0$
and have used that $\liminf_{n\to\infty} \max_x\{f_n(x)\} \geq \max_x\left \{ \lim_{n\to\infty} f_n(x)\right\}$ as long as $\lim_{n\to\infty} f_n(x)$
exists for every $x$;
and in \refE{proof2_2classes3} we have restricted the range over
which we maximize $\rho_i$, $i=0,1,2$ and interchanged the maximization
order.

By substituting \refE{Es1bound}-\refE{Es2bound} with
$0\leq\rho_1\leq\rho_0\leq\rho_2\leq 1$, the minimization in
\refE{proof2_2classes3} becomes
\begin{align}
 \min_{i=1,2} \left\{ 
           \Eo(\rho_i,\pch,Q_i) 
              + t\frac{\Es(\rho_0,\ps) - \log\gamma' }{1+\rho_0}
                 (\rho_0-\rho_i) - t \Es(\rho_0,\ps) \right\}.
         \label{eqn:proof2_2classes5}
\end{align}
We define $\gamma_0 \geq 0$ as the value satisfying
\begin{align}
  t\frac{\Es(\rho_0,\ps)-\log\gamma_0}{1+\rho_0} 
  = \frac{\Eo(\rho_2,\pch,Q_2)-\Eo(\rho_1,\pch,Q_1) }
             {\rho_2-\rho_1}.\label{eqn:proof2_2classes_slopes}
\end{align}
The existence of such $\gamma_0$ follows from the continuity of the logarithm
function.
Choosing $\gamma'=\gamma_0$ equalizes the two terms in the minimization in \refE{proof2_2classes5}, thus  maximizing the lower bound \refE{proof2_2classes3}.
As a result, substituting \refE{proof2_2classes5} into \refE{proof2_2classes3}
we obtain
\begin{align}
  \EJB
  &\geq 
%   \max_{\rho_0\in[0,1]} \left\{ \max_{\substack{0\leq\rho_1\leq\rho_0,\\\rho_0\leq\rho_2\leq 1}}
%         \min_{i=1,2} \left\{ 
%           \Eo(\rho_i,\pch,Q_i) 
%              + \frac{\Eo(\rho_2,\pch,Q_2)  - \Eo(\rho_1,\pch,Q_1)}
%                     {\rho_2-\rho_1}
%          \left( \rho_0-\rho_i \right) \right\} - t \Es(\rho_0,\ps)\right\}
%         \label{eqn:proof2_2classes6}\\
%  &=
   \max_{\rho_0\in[0,1]} \left\{ \max_{\substack{\rho_1,\rho_2 \in[0,1]:\\\rho_1\leq\rho_0\leq\rho_2}}\left\{
    \frac{\rho_2-\rho_0}
          {\rho_2-\rho_1} \Eo(\rho_1,\pch,Q_1)
   + \frac{\rho_0-\rho_1}
          {\rho_2-\rho_1} \Eo(\rho_2,\pch,Q_2) \right\}
   - t \Es(\rho_0,\ps) \right\}.
         \label{eqn:proof2_2classes7}
\end{align}

We now optimize the RHS of \refE{proof2_2classes7} over
the assignments $(Q_1,Q_2)=(Q,Q')$ and $(Q_1,Q_2)=(Q',Q)$.
By denoting by $\rho$ (resp. $\rho'$) the variable
$\rho_i$, $i=1,2$, associated to $Q$ (resp. $Q'$)
and defining $\lambda$ such that $\lambda\rho+(1-\lambda)\rho'=\rho_0$,
the optimal assignment leads to
\begin{align}
  \EJB
  &\geq 
   \max_{\rho_0\in[0,1]} \left\{
   \max_{\substack{\rho,\rho',\lambda\in[0,1]:\\ \lambda\rho+(1-\lambda)\rho'=\rho_0}}
   \Bigl\{
    \lambda \Eo(\rho,\pch,Q)
   + (1-\lambda) \Eo(\rho',\pch,Q') \Bigr\}
   - t \Es(\rho_0,\ps) \right\}.
         \label{eqn:proof2_2classes8}
\end{align}

%If we denote $\lambda \triangleq \frac{\rho_2-\rho_0}{\rho_2-\rho_1}$,
%such that $\frac{\rho_0-\rho_1}{\rho_2-\rho_1} \Eo(\rho_2,\pch,Q_2)=1-\lambda$,
%it follows that $\lambda\rho_1 + (1-\lambda)\rho_2=\rho_0$.

Theorem~\ref{thm:Cs2classes} follows from \refE{proof2_2classes8}
by noting that~\cite[Th. 5.6]{convex-rockafellar}
\begin{align}
\bar\Eo(\rho_0,\pch, \{Q,Q'\})
  = \max_{\substack{\rho,\rho',\lambda\in[0,1]:\\ \lambda\rho+(1-\lambda)\rho'=\rho_0}}
    \Bigl\{ \lambda \Eo(\rho,\pch, Q) + (1-\lambda) \Eo(\rho',\pch, Q') \Bigr\}.
         \label{eqn:proof2_2classes9}
\end{align}

A two-class partition achieving the bound in Theorem~\ref{thm:Cs2classes}
is given by \refE{A1def}-\refE{A2def}, with $\gamma = \min(1,\gamma_0^{\star})$
where $\gamma_0^{\star}$ is computed from \refE{proof2_2classes_slopes}
for the values of $\rho_0^{\star}$, $\rho_1^{\star}$, $\rho_2^{\star}$
optimizing \refE{proof2_2classes7} and the assignment
$(Q_1^{\star},Q_2^{\star})$ which leads to \refE{proof2_2classes8}.

\section{Proof of Theorem~\ref{thm:EnsembleTightness}}\label{proof_thm_ub}

%%%%%%%%%%%%%%%%%%%%%%%%%%%%%%%%%%%%%%%%%%%%%%%%%%%%%%%%%%%%%%%%%%%%

%%%%%%%%%%%%%%%%%%%%%%%%%%%%%%%%%%%%%%%%%%%%%%%%%%%%%%%%%%%%%%%%%%%%%%%%%%%%
%\subsection{Proof of Theorem \ref{thm_ub}} \label{proof_thm_ub}
%\begin{proof}

Before proving the result, we give some definitions that ease the exposition.
Let $\Acal$ be an arbitrary non-empty discrete set.  We denote the set of all probability distributions %over the elements of
on  $\Acal$ by $\Dcal(\Acal)$ and the set of types in $\Acal^n$ by $\mathcal{D}_n(\Acal)$. We further denote by $\Tcal(\Psf_{XY})$ the type-class of sequences $(\x,\y)$ with joint type $\Psf_{XY}$.
The set $\Lcal_{n}(P_{XY})$ is given by
\begin{align}
\Lcal_{n}(P_{XY}) \triangleq \Bigl\{  \bar{\Psf}_{XY} \in \Dcal_n(\Xa\times \Ya): \bar{\Psf}_Y=P_Y, \EE\bigl[\log\pch(\bar{Y}|\bar{X})\bigr]  \geq \EE\bigl[\log\pch(Y|X)\bigr] \Bigr\}, \label{asymptotic_set}
\end{align}
where $(\bar{X},\bar{Y})\sim \bar{\Psf}_{XY}$ and $(X,Y)\sim P_{XY}$, and $P_Y$ denotes the marginal distribution of $P_{XY}$. 
Here, and throughout this appendix, we indicate that $\A$ is distributed according to the distribution $P_{\A}$ by writing $\A\sim P_{\A}$. 
Analogously, we define the set $\Lcal(P_{XY})$ as
\begin{align}
\Lcal(P_{XY}) &\triangleq \Bigl\{ \bar{P}_{XY} \in \Dcal(\Xa\times \Ya):
 \bar{P}_Y=P_Y, %\notag\\&\;\qquad
 \EE\bigl[\log\pch(\bar{Y}|\bar{X})\bigr]  \geq \EE\bigl[\log \pch(Y|X)\bigr]   \Bigr\}, \label{asymptotic_set2}
\end{align}
with $(\bar{X},\bar{Y})\sim \bar{P}_{XY}$ and $(X,Y)\sim P_{XY}$.

Extending~\cite[Th. 1]{Sca12} to source-channel coding, we find that
\begin{equation}
\bar{\epsilon} \geq \frac{1}{4}\sum_{i=1}^{N'_k}\sum_{\v\in \Tcal_i} \ps(\v)\EE \Biggl[\min\biggl\{1,\sum_{\bar{\v}\in  \Tcal_i}
\Pr\Bigl\{{\Ps(\bar{\v})\Pch(\Y| \bar{\X}_i) \geq \Ps(\v)\Pch(\Y|\X_i)}
\bigl |\X_i \Y\Bigr\} \biggr\}\Biggr],
\end{equation}
where $(\X_i,\Y) \sim \Px_i \times \Pch$ and $\bar{\X}_i \sim \Px_i$.
Here we have lower-bounded $\bar{\epsilon}$ by only considering in the inner
sum those $\bar{\v}$ that are in the source type class $\Tcal_i$, $i=1, \dotsc, N'_k$.

We rewrite this bound in terms of summations over types with
\begin{align}\label{dot_eq}
\bar{\epsilon} &\geq \frac{1}{4}\sum_{i=1}^{N'_k} \sum_{\Psf_{XY}}
  \Pr\left\{\V\in \Tcal_i\right\} \Pr\bigl\{(\X_i,\Y)\in \Tcal(\Psf_{XY}) \bigr\}
\notag\\
& \qquad\qquad\times  
\min\left\{ 1, \sum_{\bar{\Psf}_{XY}\in \Lcal_{n} (\Psf_{XY})} \bigl|\Tcal_i\bigr|\Pr\Bigl\{(\bar{\X}_i,\y)\in \Tcal(\bar{\Psf}_{XY})\;\big|\;\y\in \bar{\Psf}_Y\Bigr\} \right\}, \end{align}
where $\V \sim \Ps$.

Applying \cite[Lemma 2.3]{CsisKorn11} and \cite[Lemma 2.6]{CsisKorn11},
we obtain 
\begin{align}
 \bar{\epsilon}  
&\geq \sum_{i=1}^{N'_k} \sum_{\Psf_{XY}}   
\exp\Bigl(-k D(\Psf_i\| P)-n D(\Psf_{XY} \| Q_i\times \pch)+\delta'_{k,n} -\log 4\Bigr)\notag\\
%%& \times \exp\Biggl(-\biggl[-\log \sum_{\bar{\Psf}_{XY}\in \Lcal_{n} (\Psf_{XY}, \Psf_i, \Psf_i)} \exp\Bigl(k H(V_i)-nD(\bar{\Psf}_{XY} \| Q_k^{(i)}\times \bar{\Psf}_Y)+\delta'_{k,n}\Bigr)\biggr]^{+}\Biggr)\label{step_types_prop3}
& \qquad\qquad\times \min\left\{1,\sum_{\bar{\Psf}_{XY}\in \Lcal_{n} (\Psf_{XY})} \exp\Bigl(k H(V_i)-nD(\bar{\Psf}_{XY} \| Q_i\times \bar{\Psf}_Y)+\delta'_{k,n}\Bigr)\right\},\label{step_types_prop3}
\end{align}
where $V_i\sim \Psf_i$ and $\delta'_{k,n}\triangleq \log(k+1)^{-|\Va|}(n+1)^{-|\Xa||\Ya|}$.

The error probability can be further bounded by keeping only the leading exponential term in each summation in \eqref{step_types_prop3}. Taking logarithms on both sides of \eqref{step_types_prop3}, multiplying the result by $-\frac{1}{n}$, and using the notation $[x]^{+}=\max(x,0)$
%invoke the  continuity of the function $[x]^{+}=\max(x,0)$ 
we obtain 
\begin{align}
-\frac{\log\bar{\epsilon}}{n} 
%& \leq   
%\min_{i=1, \dotsc, N'_k}\min_{\Psf_{XY}}  \Biggl\{ \frac{k}{n} D(\Psf_i\| P)+ D(\Psf_{XY} \| Q_k^{(i)}\times \pch) -\frac{\delta'_{k,n}}{n}-\frac{\log 4}{n}\notag\\
%&\qquad\quad+\left[\min_{\bar{\Psf}_{XY}\in \Lcal_{n} (\Psf_{XY}, \Psf_i, \Psf_i)} D(\bar{\Psf}_{XY} \| Q_k^{(i)}\times \bar{\Psf}_Y)-\frac{k}{n}H(V_i)-\frac{\delta'_{k,n}}{n} \right]^{+}\Biggr\}\label{step_types_prop5}\\
&\leq \min_{i=1, \dotsc, N'_k}\min_{\Psf_{XY}} \min_{\bar{\Psf}_{XY}\in \Lcal_{n} (\Psf_{XY})} \Biggl\{ \frac{k}{n} D(\Psf_i\| P)+ D(\Psf_{XY} \| Q_i\times \pch)\notag\\
&\qquad\qquad\qquad\qquad\qquad\qquad\;\;\;\; +\left[ D(\bar{\Psf}_{XY} \| Q_i\times \bar{\Psf}_Y)-\frac{k}{n}H(V_i) \right]^{+}\Biggr\}-\frac{\delta_{k,n}}{n},\label{step_types_prop6}
\end{align}
where we define $\delta_{k,n}\triangleq 2\delta_{k,n}'+\log 4$. Here
we use that $[nx]^{+}=n[x]^{+}$, for $n>0$, 
that $[x]^+=\max(0,x)$ is monotonically non-decreasing,  and that $[x+a]^{+}\leq [x]^{+}+a$, $a>0$.

Any distribution in $\Dcal(\Acal)$ can be written as  the limit %approximated arbitrarily well by
of a sequence of types in $\Dcal_n(\Acal)$~\cite[Sec. IV]{Csis98}.
%in $\Dcal_n(\Acal)$ for sufficiently large $n$~\cite[Sec. IV]{Csis98}. 
Hence, the uniform continuity of  $D(A\|B)$ over the pair $(A,B)$ % \cite{Csis80}  %divergence %A continuity argument shows 
ensures that for every $P_{XY}$,
and every $\xi_1>0$, there exists a sufficiently large $n$ such that
\begin{align}
-\frac{\log\bar{\epsilon}}{n} 
& \leq \min_{i=1, \dotsc, N'_k} \min_{ P_{XY}}
\min_{\bar{P}_{XY} \in \Lcal(P_{XY})}
 \Biggl\{\frac{k}{n}D(\Psf_i\| P) + D(P_{XY} \| Q_i \times \pch)\notag\\ 
 & \qquad\qquad\qquad\qquad\qquad\qquad\quad+\left[D(\bar{P}_{XY} \| Q_i \times \bar{P}_Y)-\frac{k}{n}H(V_i)\right]^{+}\Biggr\}-\frac{\delta_{k,n}}{n}+\xi_1, \label{err_exp_iid_prelim1}
\end{align}
where  we have replaced $\Lcal_{n} (\Psf_{XY})$ by $\Lcal(P_{XY})$, and used that $[x+a]^{+}\leq [x]^{+}+a$, $a>0$.
%where $V'\sim P$,  and

It follows from \cite[Th. 4]{Sca12} that
%we have the following exponent identity
\begin{multline}\label{eqn:jon_thm4}
\min_{ P_{XY}}
\min_{\bar{P}_{XY} \in \Lcal(P_{XY})}
 \left\{ D(P_{XY} \| \px\times \pch)+[D(\bar{P}_{XY} \| \px\times \bar{P}_Y)-R]^{+}\right\} \\ = 
\max_{\rho\in[0,1]} \bigl\{ \Eo(\rho,\pch,\px)-\rho R \bigr\},
\end{multline}
so  \eqref{err_exp_iid_prelim1} is equivalent to 
\begin{align}
-\frac{\log\bar{\epsilon}}{n} & \leq \min_{i=1, \dotsc, N'_k} 
 \Biggl\{\frac{k}{n}D(\Psf_i\| P) + 
   \max_{\rho\in[0,1]} \left\{ \Eo(\rho,\pch, Q_i)-\rho \frac{k}{n} H(V_i) \right\}
 \Biggr\} -\frac{\delta_{k,n}}{n}+\xi_1. \label{err_exp_iid_prelim12}
\end{align}

Maximizing \eqref{err_exp_iid_prelim12}  over $Q_i\in\mathcal{Q}$ for each $i=1,\ldots,N'_k$ yields
%We now upper-bound the RHS of \eqref{err_exp_iid_prelim12} by maximizing over 
%$P_X\in\Qc$ for each $i=1, \dotsc, N'_k$,
%and apply \ref{eqn:jon_thm4} to obtain
\begin{align}
-\frac{\log\bar{\epsilon}}{n} & \leq \min_{i=1, \dotsc, N'_k} 
 \Biggl\{\frac{k}{n}D(\Psf_i\| P) + 
   \max_{\rho\in[0,1]}  \left\{ \Eo(\rho,\pch,\Qc)-\rho \frac{k}{n} H(V_i) \right\}
 \Biggr\} -\frac{\delta_{k,n}}{n}+\xi_1. \label{err_exp_iid_prelim2}
\end{align}

By taking $n$ to be sufficiently large in the outer bracketed term of  \eqref{err_exp_iid_prelim2}, 
we obtain for $\xi_2>0$ that
\begin{align}
-\frac{\log\bar{\epsilon}}{n} & \leq \min_{i=1, \dotsc, N'_k} 
 \Biggl\{tD(\Psf_i\| P) + 
   \max_{\rho\in[0,1]}  \left\{ \Eo(\rho,\pch,\Qc)-\rho t H(V_i) \right\}
 \Biggr\} -\frac{\delta_{k,n}}{n}+\xi_1+\xi_2. \label{err_exp_iid_prelim2b}
\end{align}

Using now the uniform continuity of the RHS of \eqref{err_exp_iid_prelim2b}
%\begin{equation}
%\max_{\rho\in[0,1]}  \left\{ \Eo(\rho,\pch,\Qc)-\rho R \right\}
%\end{equation}
as a function of  $\Psf_i$ \cite[p. 323]{Csis80}
%$R$ \cite{CsisKorn11}, 
and that any distribution in $\mathcal{D}(\Va)$ can be written as the limit of a sequence of source types in $k$,
 it follows that for every $\xi_3>0$ there exists a sufficiently large $n$ such that
%\begin{align}
%-\frac{\log\bar{\epsilon}}{n}
%\leq & \min_{P'} \Biggl\{\frac{k}{n}D(P'\| P)+
%\max_{\rho\in[0,1]} \left\{ \Eo(\rho,\pch,\Qc)-\rho \frac{k}{n}H(V') \right\} \Biggr\}-\frac{\delta_{k,n}}{n} +\xi_1+\xi_2+\xi_3,
%\end{align}
\begin{align}
-\frac{\log\bar{\epsilon}}{n}
\leq & \min_{P'} \Biggl\{t D(P'\| P)+
\max_{\rho\in[0,1]} \left\{ \Eo(\rho,\pch,\Qc)-\rho t H(V') \right\} \Biggr\}-\frac{\delta_{k,n}}{n} +\xi_1+\xi_2+\xi_3,
\end{align}
where $V'\sim P'$ .
%
%
%the minimization over source types
%in \eqref{err_exp_iid_prelim2} can be replaced in the limit
%by a minimization over probability distributions.
%Finally, 
By taking the limit superior in $n$, 
%subject
%to the restriction that $\lim_{n\to \infty}\frac{k}{n}= t$,
this  becomes
\begin{align}
\limsup_{n\to \infty} 
-\frac{\log\bar{\epsilon}}{n}
\leq & \min_{P'} \Bigl\{tD(P'\| P)+
\max_{\rho\in[0,1]} \left\{ \Eo(\rho,\pch,\Qc)-\rho t H(V') \right\} \Bigr\}\label{quasi_last-iidN}+\xi_1+\xi_2+\xi_3\\
= & \min_{0\leq R \leq t \log |\Vcal|} \Bigl\{ te\left(\frac{R}{t}, \ps\right)+
\max_{\rho\in[0,1]} 
 \left\{ \Eo(\rho,\pch,\Qc)-\rho R \right\} \Bigr\} +\xi_1 +\xi_2+\xi_3 \label{last-iidN}
\\
= & \max_{\rho\in[0,1]} \bigl\{ \bar{\Eo}(\rho, \pch, \Qc)-t \Es(\rho, \ps)\bigr\}+\xi_1+\xi_2 +\xi_3,   \label{very_last-iidN}
\end{align}
where %$V'\sim P'$  in \eqref{quasi_last-iidN};
%\eqref{ccoding} follows from \cite[Cor. 3]{Sca12};
%\eqref{minmax} is a consequence of the inequality $\max\min\{\cdot\}\leq \min\max\{\cdot\}$; 
\eqref{last-iidN} follows from the definition of the source reliability function \cite[eq. (7)]{Csis80}
with $R=t H(V')$;
and \eqref{very_last-iidN} can be proved by the same methods that relate~\eqref{Csiszar_exponent} and~\eqref{Concave_Hull_exponent}.
Finally, letting $\xi_1$, $\xi_2$ and $\xi_3$ tend to
zero from above yields the desired result.

\bibliographystyle{IEEEtran}
\bibliography{bib/references-2}

\end{document}